%% file: main.tex
\documentclass{llncs}

\usepackage{makeidx}  

\usepackage{mathrsfs}
\usepackage{enumitem}
\usepackage{stmaryrd}
\usepackage{makeidx}  
\usepackage{booktabs}
\usepackage{graphicx}

\usepackage{csquotes}
\usepackage{mathtools}
\usepackage{amssymb}
\usepackage{amsmath}

\usepackage{hyperref}

\usepackage{tikz}
\usetikzlibrary{arrows,shapes.multipart,backgrounds,automata,positioning}
\usepackage{subfig} 
\SetLabelAlign{CenterWithParen}{\hfil(\makebox[1.0em]{#1})\hfil}

\usepackage{ifthen}
\newcommand{\forLoop}[5][1]{
  \setcounter{#4}{#2}
  \ifthenelse{\value{#4}>#3}{}
  { 
    #5
    \addtocounter{#4}{#1}
    \forLoop[#1]{\value{#4}}{#3}{#4}{#5}
  }
}

\usepackage{myalg}

\title{Model Checking Gene Regulatory Networks
\thanks{This research was supported by the European Research Council (ERC) under grant 267989 (QUAREM),
the Austrian Science Fund (FWF) under grants S11402-N23 (RiSE) and Z211-N23 (Wittgenstein Award),
the European Union's 
SAGE grant agreement no. 618091, ERC Advanced Grant ERC-2009-AdG-250152, the People Programme (Marie Curie Actions) of the European Union's Seventh Framework Programme (FP7/2007-2013) under REA grant agreement no. 291734, and the SNSF Early Postdoc.Mobility Fellowship, the grant number P2EZP2\_148797.
}
}
\titlerunning{Model checking Gene Regulatory Networks}  
\author{Mirco Giacobbe${}^{*}$ \and C\u{a}lin C. Guet${}^{*}$ \and Ashutosh Gupta${}^{*\dagger}$ \and \\
  Thomas A. Henzinger${}^{*}$ \and Tiago Paix\~{a}o${}^{*}$ \and Tatjana Petrov${}^{*}$}
\authorrunning{Giacobbe et al.} 
\tocauthor{Mirco Giacobbe, C\u{a}lin C. Guet, Ashutosh Gupta, Thomas A. Henzinger, Tiago Paix\~{a}o, Tatjana Petrov}
\institute{IST Austria, Austria${}^{*}$ \quad TIFR, India${}^{\dagger}$\\
}

\begin{document}

\maketitle              
\pagenumbering{arabic}

\input{tex_parts/macros.tex}

\begin{abstract}
\input{tex_parts/abstract.tex}

\end{abstract}

\section{Introduction}
\label{sec:intro}
\input{tex_parts/intro.tex}
\subsection{Motivating example}
\label{sec:example}
\input{tex_parts/example.tex}

\section{Preliminaries}
\label{sec:prelim}
\input{tex_parts/preliminaries.tex}

%

\section{Algorithm for parameter synthesis}
\label{sec:algo}
\input{tex_parts/algorithm.tex}
\section{Computing Robustness}
\label{sec:robustness}
\input{tex_parts/robustness.tex}


\section{Experimental results}
\label{sec:experiments}

\input{tex_parts/experiments.tex}
\section{Conclusion and discussion}
\label{sec:conclusion}
\input{tex_parts/conclusion.tex}
%

%
%
\bibliographystyle{abbrv}
\bibliography{biblio}

\appendix

\input{tex_parts/mutation-model.tex}


\end{document}

%% file: tex_parts/macros.tex
\newcommand{\ashu}[1]{ {\textcolor{magenta} {Ashu: #1}} }
\newcommand{\mirco}[1]{ {\textcolor{red} {Mirco: #1}} }
\newcommand{\tanja}[1]{ {\textcolor{red} {Tanja: #1}} }
\newcommand{\tiago}[1]{ {\textcolor{blue} {Tiago: #1}} }

\def\GRNlandscape{GRN space}
\def\POPGRN{GRN-population}
\def\DETGRN{GRN-individual}

\def\run{r}

\def\tran{{\cal T}}
\def\inputG{in}
\def\detgrn{{\mathsf{detG}}}
\def\popgrn{{\mathsf{popG}}}

\def\topp{T}
\def\cons{{cons}}

\def\rob{\rho}
\def\K{M}
\def\Pmat{\mathbf{P}}
\def\Qmat{\mathbf{T}}
\def\Q{T}

\def\L{{\cal L}} 

\def\Kvec{\mathbf{M}}

\def\Wvec{\mathbf{W}}
\def\Wrandom{W}

\def\Lvec{\mathbf{L}}
\def\kvec{\mathbf{k}}
\def\zero{\mathbf{0}}
\def\one{\mathbf{1}}

\def\ivec{\mathbf{i}}
\def\jvec{\mathbf{j}}
\def\lvec{\mathbf{l}}
\def\Kvecbias{\mathbf{\tilde{K}}}

\tikzstyle{gen}=[shape=circle,draw=black!50, fill=black!20, thick, inner sep=0pt, minimum size=10mm]
\tikzstyle{act}=[->,thick]
\tikzstyle{rep}=[-|,thick]
\tikzstyle{seqstate}=[shape=circle,draw=black!50, fill=white, thick, inner sep=0pt, minimum size=7mm]
\newcommand{\infseq}[5]{
\begin{tikzpicture}
\node[seqstate] (s1) {#1}; 
\node[seqstate] (s2) [right of=s1,xshift=1cm] {#2} 
  edge [<-,thick] (s1) {}; 
\node[seqstate] (s3) [right of=s2,draw=none,xshift=1cm] {~\dots} 
  edge [<-,dashed] (s2) {}; 
\node[seqstate] (s4) [right of=s3,xshift=1cm] {#3} 
  edge [<-,dashed] (s3) {}; 
\node[seqstate] (s5) [right of=s4,xshift=1cm] {#4} 
  edge [<-,thick] (s4) {}; 
\node[seqstate,draw=none] (s6) [right of=s5,xshift=1cm] {~\dots ~#5} 
  edge [<-,dashed] (s5) {}; 
\end{tikzpicture}
}

\def\len{l}

\newcommand{\corref}[1]{Cor.~\ref{#1}}
\newcommand{\figref}[1]{Fig.~\ref{#1}}
\newcommand{\tableref}[1]{Table~\ref{#1}}
\newcommand{\secref}[1]{Section~\ref{#1}}
\newcommand{\thmref}[1]{Thm.~\ref{#1}}
\newcommand{\dfnref}[1]{Dfn.~\ref{#1}}
\newcommand{\lemref}[1]{Lem.~\ref{#1}}
\newcommand{\exref}[1]{Ex.~\ref{#1}}



\newcommand{\union}{\cup}

\newcommand{\trProb}{{T}}
\newcommand{\ade}{\mathtt{A}}
\newcommand{\tim}{\mathtt{T}}
\newcommand{\cyt}{\mathtt{C}}
\newcommand{\gua}{\mathtt{G}}
\newcommand{\avec}{\mathbf{a}}
\newcommand{\Mut}{\mathtt{Mut}}

\newcommand{\weightSet}{{W}}
\newcommand{\wmax}{w^{max}}
\newcommand{\Lan}{{\cal L}}
\newcommand{\pr}{p}
\newcommand{\nuc}{a}
\newcommand{\Nuc}{A}
\newcommand{\B}{\mathbb{B}}
\newcommand{\R}{\mathbb{R}}
\newcommand{\N}{\mathbb{N}}

\newcommand{\weightA}{\weight_A}
\newcommand{\weightB}{\weight_B}
\newcommand{\powerset}{\mathcal{P}}

\newcommand{\wVec}{\textbf{\weight}}
\newcommand{\kVec}{\textbf{k}}
\newcommand{\threshA}{t_A}
\newcommand{\threshB}{t_B}
\newcommand{\inA}{i_A}
\newcommand{\inB}{i_B}

\newcommand{\piVec}{\boldsymbol{\pi}}
\newcommand{\prob}{p}
\newcommand{\probVec}{\mathbf{\prob}}
\newcommand{\lengthVec}{\mathbf{\length}}

\newcommand{\grn}{\mathcal{G}}
\newcommand{\grnP}{\mathcal{Z}}
\newcommand{\grnset}{\Gamma}
\newcommand{\gens}{G}
\newcommand{\acts}{\tikz[baseline=-0.5ex] \draw[->] (0ex,0ex) -- (2ex,0ex);}
\newcommand{\reps}{\tikz[baseline=-0.5ex] \draw[-|] (0ex,0ex) -- (2ex,0ex);}
\newcommand{\wagner}{\mathcal{W}}
\newcommand{\weight}{w}
\newcommand{\thresh}{t}
\newcommand{\langstat}{\mathcal{L}}
\newcommand{\init}{I_0}
\newcommand{\invar}{I_*}
\newcommand{\gen}{g}
\newcommand{\genvar}[1]{#1}
\newcommand{\weightvar}[1]{\weight({#1})}
\newcommand{\allVars}{V}
\newcommand{\wVar}{v}
\newcommand{\ltlnext}{\mathcal{X}}
\newcommand{\ltlglobally}{\Box}
\newcommand{\ltlfinally}{\Diamond}
\newcommand{\ltluntil}{\mathcal{U}}
\newcommand{\ctlexists}{\textbf{E}}
\newcommand{\prop}{\varphi}
\newcommand{\length}{l}
\newcommand{\mutable}{\mathcal{M}}
\newcommand{\state}{\mathbf{\sigma}}
\newcommand{\states}{\Sigma}
\newcommand{\stateVar}{\mathbf{v}}
\newcommand{\trans}{\tau}

\newcommand{\mutationVec}{\mathbf{k}}

\newcommand{\bound}{n}
\newcommand{\enc}[1]{[[#1]]}
\newcommand{\transSys}{\emph{Tr}}
\newcommand{\transVar}{\Phi}
\newcommand{\transenc}[3]{\transVar(#1,#2,#3)}
\newcommand{\genenc}{s}
\newcommand{\pregenenc}{\emph{Pre}}
\newcommand{\postgenenc}{\emph{Post}}
\newcommand{\ite}{\emph{ite}}
\newcommand{\witness}{\neg{\prop}}
\newcommand{\gene}{g}
\newcommand{\geneN}{d}

\newcommand{\diag}[1]{\emph{diag}(#1)}

\def\Robustness{\mbox{Robustness}}
\def\Property{\phi}
\def\Perturbation{\mbox{Perturbations}}
\def\mutation{\mbox{mut}}

\def\outcomes{\Omega}
\def\sigmalgebra{\mathcal{F}}
\def\probmeasure{P}

\def\flipprob{{p_f}}

\newtheorem{defn}{Definition}

\newcommand{\Def}[1]{Def. \ref{#1}}
\newcommand{\Fig}[1]{Fig. \ref{#1}}
\newcommand{\Eq}[1]{Eq. \ref{#1}}
\newcommand{\Ex}[1]{Ex. \ref{#1}}

\def\wagners{Wagner's}

\def\wI{0.75}
\def\wA{1}
\def\wB{1}
\def\tA{0.5}
\def\tB{0.5}
\def\geneA{a}
\def\geneB{b}

\def\grnU{\grn}
\def\grnR{\grn^r}

\def\fo{\trans}
\def\ft{\trans^r}

\def\tc{.75}
\def\tb{.25}
\def\wa{0.5}
\def\wb{.5}
\def\la{5}
\def\lb{5}
\def\fractc{\frac{3}{4}}
\def\fractb{\frac{1}{4}}
\def\fracwa{\frac{1}{2}}
\def\fracwb{\frac{1}{2}}
\def\fracmutwb{\frac{2}{10}}

\def\polyhedra{polyhedra}
\def\linear{linear}
\def\machine{\mirco{machine here}}
\def\mathsat{\textsc{MathSAT5}}

\newcommand{\algGenCons}{\textsc{GenCons}\xspace}
\newcommand{\algGenConsRec}{\textsc{GenConsRec}\xspace}



%% file: tex_parts/abstract.tex
The behaviour of gene regulatory networks (GRNs) is typically
analysed using simulation-based statistical testing-like methods.
In this paper, we demonstrate that we can replace this
approach by a formal verification-like method that gives higher
assurance and scalability. 
We focus on Wagner's weighted GRN model with varying weights, 
which is used in evolutionary biology. 
In the model, weight parameters represent the gene interaction strength
that may change due to genetic mutations.
For a property of interest, we synthesise the constraints
over the parameter space that represent the set of GRNs satisfying 
the property.
We experimentally show that our parameter synthesis
procedure computes the mutational robustness of GRNs --an important
problem of interest in evolutionary biology-- more efficiently 
than the classical simulation method.
We specify the property in linear temporal logics.
We employ symbolic bounded model checking and SMT solving to compute
the space of GRNs that satisfy the property, which amounts to
synthesizing a set of linear constraints on the weights.



%% file: tex_parts/intro.tex
%
Gene regulatory networks (GRNs) are one of the most prevalent and fundamental
type of biological networks whose main actors are genes regulating
other genes.
A topology of a GRN is represented by a graph of interactions among a
finite set of genes, where nodes represent genes, and edges denote the
type of regulation (activation or repression) between the genes, if any.
%
%
In~\cite{wagner_does_1996}, Wagner introduced a simple but useful
model for GRNs that captures important features of GRNs.
In the model, a system state specifies the activity of each gene as a
Boolean value.
The system is executed in discrete time steps, and all gene values are
synchronously and deterministically updated: a gene active at time $n$
affects the value of its neighbouring genes at time $n+1$.  
This effect is modelled through two kinds of parameters: {\it
threshold} parameters assigned to each gene, which specify the
strength necessary to sustain the gene's activity, and {\it weight}
parameters assigned to pairs of genes, which denote the strength of
their directed effect.

%
Some properties of GRNs can be expressed in linear temporal
logic (LTL)(such as reaching a steady-state), where atomic
propositions are modelled by gene values.
A single GRN may or may not satisfy a property of interest.
Biologists are often interested in the behavior of \emph{populations
of GRNs}, and in presence of environmental perturbations.
For example, the parameters of GRNs from a population may change from
one generation to another due to mutations, and the distribution over
the different GRNs in a population changes accordingly.
We refer to the set of GRNs obtained by varying parameters on
a fixed topology as {\em \GRNlandscape}.
For a given population of GRNs instantiated from a \GRNlandscape, 
typical quantities of interest refer to the long-run average
behavior.
For example, \emph{robustness} refers to the averaged satisfiability of
the property within a population of GRNs, after an extended number of
generations. 
%
%
In this context, Wagner's model of GRN has been used to show
that mutational robustness can gradually evolve in
GRNs~\cite{Wagner2006_robustness}, that sexual reproduction can
enhance robustness to recombination~\cite{Burch2006_robusteness},
or to predict the phenotypic effect of mutations~\cite{mccarthy2003}.
%
The computational analysis used in these studies relies on explicitly
executing GRNs, with the purpose of checking if they satisfy the
property.
Then, in order to compute the robustness of a population of GRNs,
the satisfaction check must be performed repeatedly for many
different GRNs.
In other words, robustness is estimated by statistically sampling GRNs from
the \GRNlandscape\;and executing each of them until the property is (dis)proven.
In this work, we pursue formal analysis of Wagner's GRNs which allows
to avoid repeated executions of GRNs,  
and to compute mutational robustness with higher assurance and scalability.

%
In this paper, we present a novel method for synthesizing the space of parameters
which characterize GRNs that satisfy a given property.
These constraints eliminate the need of explicitly executing the GRN
to check the satisfaction of the property.
%
Importantly, the synthesized parameter constraints allow to efficiently answer questions 
that are very difficult or impossible to answer by simulation, 
e.g. emptiness check or parameter sensitivity analysis.
%
%
In this work, we chose to demonstrate how the synthesized constraints can be used to 
compute the robustness of a population of GRNs with respect to genetic mutations.
%
Since constraint evaluation is usually faster than executing a GRN, the
constraints pre-computation enables faster computation of robustness.
This further allows to compute the robustness with higher precision, within the same computational time.
Moreover, it sometimes becomes possible to replace the
statistical sampling with the exact computation of robustness. 
%



%

In our method, for a given \GRNlandscape\;and LTL property, we used
SMT solving and bounded model checking to generate a set of
constraints such that a GRN satisfies the LTL property if and only if
its weight parameters satisfy the constraints.
The key insight in this method is that the obtained constraints are
complex Boolean combinations of linear inequalities.
Solving linear constraints has been the focus of both industry and
academia for some time.
However, the technology for solving linear constraints with Boolean
structure, namely SMT solving, has matured only in the last
decade~\cite{barrett20136}.
This technology has enabled us to successfully apply an SMT solver to
generate the desired constraints.

We have {built} a tool which computes the constraints for a given
\GRNlandscape\; and a property expressed in a fragment of LTL.
In order to demonstrate the effectiveness of our method, we computed
the robustness of five GRNs listed in~\cite{cardellimorphisms}, and
for three GRNs known to exhibit oscillatory behavior.
We first synthesized the constraints and then we used them to estimate
robustness based on statistical sampling of GRNs from the GRN space.
Then, in order to compare the performance with the simulation-based
methods, we implemented the approximate computation of robustness,
where the satisfiability of the property is verified by executing the
GRNs explicitly.
The results show that in six out of eight tested networks, the
pre-computation of constraints provides greater efficiency, performing
up to three times faster than the simulation method.

\noindent {\textbf{Related Work}} 
Formal verification techniques are already used for aiding various
aspects of biological
research~\cite{danos2004formal,fisher2007executable,kwiatkowska2008using,jha2009bayesian,yordanov2013smt}. In particular, the robustness of models of biochemical systems with
respect to temporal properties has been studied
\cite{mateescu2011ctrl,batt2007model,Rizk2009_robustness}. 
%
Our work is, to the best of our knowledge, the first application of
formal verification to studying the evolution of mutational robustness
of gene regulatory networks and, as such, it opens up a novel
application area for the formal verification community.
As previously discussed, with respect to related studies in
evolutionary biology, our method can offer a higher degree of
assurance, more accuracy, and better scalability than the traditional,
simulation-based approaches.
In addition, while the mutational robustness has been studied only for invariant 
properties, our method allows to compute the mutational robustness 
for non-trivial temporal properties that are
expressible in LTL, such as bistability or oscillations between gene
states.

%% file: tex_parts/example.tex
\input{./fig/fig-bi-stable_exp}
In the following, we will illustrate the main idea of the paper on an
example of a \GRNlandscape\;$\topp$ generated from the GRN network
shown in \figref{fig:ab}(a).
Two genes $A$ and $B$ inhibit each other, and both genes have a
self-activating loop.
The parameters $(\inA,\inB)$ represent constant inputs, which we
assume to be regulated by some other genes that are not presented in
the figure.
Each of the genes is assigned a threshold value ($\threshA$,
$\threshB$), and each edge is assigned a weight
($\weight_{AA}$,$\weight_{AB}$, $\weight_{BA}$, $\weight_{BB}$).
The dynamics of a \DETGRN\;chosen from $\topp$ depends on these
parameters.
Genes are in either  active or inactive state, which we represent
with Boolean variables. 
For a given initial state of all genes, and for fixed values of
weights and thresholds, the values of all genes evolve
deterministically in discrete time-steps that are synchronized over
all genes.
%
Let $\geneA$ (resp. $\geneB$) be the Boolean variable representing the
activity of gene $A$ (resp. $B$). 
We denote a GRN state by a pair $(\geneA,\geneB)$.
Let $\fo$ be the function that governs the dynamics of
$\grnU$ (see~Def.~\ref{def:grn-sem}):
\begin{align*} 
\fo(\geneA,\geneB) &= (\inA +\geneA \weight_{AA}-\geneB \weight_{BA} > \threshA,\; \inB+\geneB\weight_{BB}-\geneA \weight_{AB} > \threshB ) 
\end{align*}
The next state of a gene is the result of arithmetically adding  
the influence from other active genes.

The topology of mutually inhibiting pair of genes is known to be
bistable: whenever one gene is highly expressed and the other is
barely expressed, the system remains stable
\cite{gardner2000construction,Rizk2009_robustness}.
%
%
The bistability property can be written as the following LTL formula (see
Def.~\ref{def:ltl}):
$$
( A \land \neg B \implies \ltlglobally (A \land \neg B )  ) \land
( \neg A \land B \implies \ltlglobally (\neg A \land B )  ).
$$
Let us fix values for parameters $t_A=t_B=0.6$,
$\weight_{AB}=\weight_{BA}=\weight_{BB}=0.3$,
and $\inA=\inB = \frac{2}{3}$.
%
%
Then, we can check that a GRN is bistable by executing the GRN.
Indeed, for the given parameter values, the states $(0,1)$ and $(1,0)$
are fixed points of $\fo$.
In other words, the GRN with those parameter values have two stable
states: if they start in state $(0,1)$ (resp. $(1,0)$), they remain
there.
Now let us choose $\inA=\frac{2}{3}$, $\inB=\frac{1}{3}$.
Again, by executing the GRN, we can conclude that it does not satisfy the property:
 at state $(0,1)$, $B$ does not have sufficiently strong
activation to surpass its threshold and the system jumps to $(0,0)$.
%
%
Intuitively, since the external activation of $B$ is too small,
the phenotype has changed to a single stable state.
In general, it is not hard to inspect that the bistability property will be met by any choice of parameters satisfying the following constraints:
\begin{align} 
\{\inA-\weight_{BA}\leq \threshA,
\;\inA+\weight_{AA}>\threshA,\;\inB-\weight_{AB}\leq \threshB,
\;\inB+\weight_{BB}>\threshB\}.
\label{eq:cons}
\end{align}

Let's now suppose that we want to compute the robustness of $\topp$ in
presence of variations on edges due to mutations.
Then, for each possible value of parameters, one needs to verify if the respective
GRN satisfies the property. 
%
%
Using the constraints~\eqref{eq:cons}, one may verify GRNs without
executing them.
%

%

%
%

Our method automatizes this idea to any given GRN topology and any property specified in LTL.
We first encode $\topp$ as a parametrized labelled
transition system, partly shown in ~\figref{fig:ab}(b).
%
Our implementation does not explicitly construct this transition system, nor executes the GRNs (the implementation is described in~\secref{sec:implementation}).
%
Then, we apply symbolic model checking to compute the constraints which
represent the space of GRN's from $\topp$ that satisfy the bi-stability property.

To illustrate the scalability of our method in comparison with the standard methods, in~\figref{fig:ab}(c), we compare the performance of computing the mutational robustness with and without precomputing the constraints (referred to as \emph{evaluation} and \emph{execution} method respectively).
%
We choose a mutation model such that each
parameter takes 13 possible values distributed according to the binomial distribution (see
Appendix for more details on the mutation model).
We estimate the robustness value by statistical sampling of the possible parameter values.
For a small number of samples, our method is slower because we spend
extra time in computing the constraints.
However, more samples may be necessary for achieving the desired precision. 
As the number of samples increases, our method becomes faster, because each evaluation of the constraints is two times faster than checking bistability by executing \DETGRN s.
For $1.2\times 10^5$ many simulations, execution and evaluation
methods take same total time, and the robustness value estimated from these many
samples lies 
in the interval $(0.8871,0.8907)$ with $95\%$ confidence.
Hence, for this GRN, if one needs better precision for the robustness
value, our method is preferred.

One may think that for this example, we may compute exact robustness because 
the number of parameter values is only $13^6$ (four weights and two inputs).
For simplicity of illustration, we chose this example, and we later present
examples with a much larger space of parameters, for which 
exact computation of robustness is infeasible.

%% file: fig/fig-bi-stable_exp.tex
\begin{figure}[t]
\centering
\begin{minipage}{0.30\linewidth}
  \vspace{-1ex}
\center
\begin{tikzpicture}[->,>=stealth',shorten >=1pt,auto,node distance=1.7cm,
  thick,main node/.style={circle,draw}]

  \node[main node] (A) {A};
  \node[main node] (B) [right of=A] {B};
  \node (11) [above of=A,yshift=-0.7cm] {};
  \node (22) [above of=B,yshift=-0.7cm] {};
  \node (1T) [below of=A,yshift=1.2cm] {$t_A$};
  \node (2T) [below of=B,yshift=1.2cm] {$t_B$};

  \path[every node/.style={font=\sffamily\small}]
    (A) edge [bend left,-|] node[above] {$w_{AB}$} (B)
        edge [loop left] node[below] {$w_{AA}$} (A)
    (B) edge [bend left,-|] node[below] {$w_{BA}$} (A)
        edge [loop right] node[above] {$w_{BB}$} (B)
    (11) edge node [left] {$i_{A}$} (A)
    (22) edge node [right] {$i_{B}$} (B);

\end{tikzpicture}
(a)
\end{minipage}
\begin{minipage}{0.28\linewidth}
\center
  \vspace{-1ex}
\begin{tikzpicture}[->,>=stealth',shorten >=1pt,auto,node distance=1.6cm,
  thick,main node/.style={rectangle,draw}]

  \node[main node] (01) {01};
  \node[main node] (00) [right of=01] {00};
  \node[main node] (10) [below of=01] {10};
  \node[main node] (11) [right of=10] {11};

  \node (11A) [left of=11,xshift=0.9cm] {};
  \node (11B) [left of=11,xshift=0.9cm,yshift=-3mm] {};
  \node (11C) [left of=11,xshift=0.9cm,yshift= 3mm] {};

  \node (10A) [right of=10,xshift=-0.9cm] {};
  \node (10B) [right of=10,xshift=-0.9cm,yshift=-3mm] {};
  \node (10C) [right of=10,xshift=-0.9cm,yshift= 3mm] {};

  \node (00A) [below of=00,yshift=0.9cm] {};
  \node (00B) [below of=00,yshift=0.9cm,xshift=-3mm] {};
  \node (00C) [below of=00,yshift=0.9cm,xshift= 3mm] {};

  \path[every node/.style={font=\sffamily\small}]
    (01) edge [loop above] node[left] {$\alpha_0,\beta_1$} (01)
         edge [] node[above] {$\alpha_0,\bar{\beta_1}$} (00)
         edge [] node[left] {$\bar{\alpha_0},\bar{\beta_1}$} (10)
         edge [] node[right] {$\bar{\alpha_0},\beta_1$} (11)
    (10) edge [loop below] node[right] {$\alpha_1,\beta_0$} (10)
          edge [dotted,-] node[above] {} (10A)
          edge [dotted,-] node[above] {} (10B)
          edge [dotted,-] node[above] {} (10C)
    (11)  edge [dotted,-] node[above] {} (11A)
          edge [dotted,-] node[above] {} (11B)
          edge [dotted,-] node[above] {} (11C)
    (00)  edge [dotted,-] node[above] {} (00A)
          edge [dotted,-] node[above] {} (00B)
          edge [dotted,-] node[above] {} (00C)
        ;

\end{tikzpicture}\\
(b)
\end{minipage}
\begin{minipage}{0.4\linewidth}  
  \center
  \vspace{-1ex}
\includegraphics[width=0.95\textwidth]{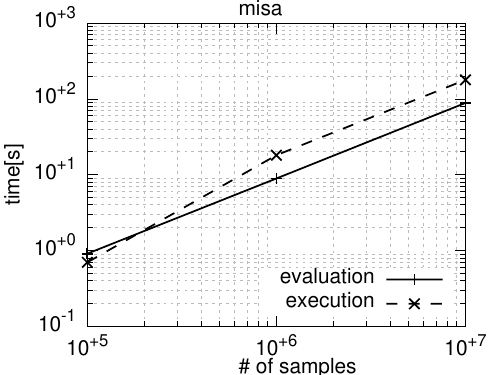} \\
(c)
\end{minipage}

\caption{
Motivating example.
a) The topology of a GRN with mutual inhibition of genes $A$ and $B$.
b) The labelled transition system where labels denote the linear
constraints which enable the transition. $\alpha_0$, $\alpha_1$,
$\beta_0$, and $\beta_1$ denote linear atomic formulas $i_A - w_{BA} \leq
t_A$, $i_B+w_{BB} > t_B$, $i_B - w_{AB} \leq t_B$, and $i_A+w_{AA} > t_A$
respectively.
c) Run-times comparison between execution method (dashed lines)
and our evaluation method (solid lines).
}
\label{fig:ab}
\end{figure}


%% file: tex_parts/preliminaries.tex
In this section, we start by defining a {\em \GRNlandscape}, 
which will serve to specify common features for GRNs from the same population.
These common features are types of gene interactions (topology), constant parameters (thresholds), and ranges of parameter values that are subject to some environmental perturbation (weights).
%
%
Then, we formally introduce a model of an individual GRN from the
\GRNlandscape\;and temporal logic to express its properties.
%

\subsection{Basic notation}

$\mathbb{R}_{\geq 0}$ (resp. $\mathbb{Q}_{\geq 0}$) is the set of
non-negative real (resp. rational) numbers.
For $m < n$, let $m..n$ denote the set of integers from $m$ to $n$.
%
%
%
With abuse of notation, we treat finite maps with ordered domain
as vectors with size of the map.

Let $\allVars$ be an infinite set of variable names.
For rationals $k_1,\dots,k_n$, a rational variable vector $\wVar =
(v_1,\dots,v_n)$, and a rational $t$, let $ k_1v_1 + \dots + k_nv_n +
t$ denote a linear term.
Let $k_1v_1 + \dots + k_nv_n + t > 0 $ and $k_1v_1 + \dots + k_nv_n +
t \geq 0$ be a strict and non-strict inequality over $\wVar$
respectively.
Let $\linear(\wVar)$ be the set of all the (non-)strict inequalities
over $\wVar$.
Let $\polyhedra(\wVar)$ be the set of all the finite conjunctions of
the elements of $linear(\wVar)$.

\subsection{\GRNlandscape}

%
 
%
The key characteristics of the behaviour of a GRN are typically summarised
by a directed graph where nodes represent genes and edges denote the
type of regulation between the genes.
A regulation edge is either {\em activation} (one gene's activity
increases the activity of the other gene) or {\em repression} (one
gene's activity decreases the activity of the other
gene)~\cite{schlitt2007current}.
In Wagner's model of a GRN, in addition to the activation types
between genes, each gene is assigned a \emph{threshold} and each edge
(pair of genes) is assigned a \emph{weight}.  
The threshold of a gene models the amount of activation level necessary
to sustain activity of the gene. 
The weight on an edge quantifies the influence of the source gene  
on destination gene of the edge.

We extend the Wagner's model by allowing a range of values for weight parameters.
%
We call our model \GRNlandscape, 
%
denoting that all GRNs instantiated from that space share the same topology, and their parameters fall into given ranges.
%
%
%
%
%
%
We assume that each gene always has some minimum level of
expression without any external influence.
In the model, this constant input is incorporated by a special
gene which is always active, and activates all other genes from the
network.
The weight on the edge between the special gene and some other gene
represents the minimum level of activation.
The minimal activation is also subject to perturbation. 


\begin{defn} [\GRNlandscape] \rm \label{def:grn_topology}
A \GRNlandscape\;is a tuple 
$$
\topp = (\gens, \gene_{\inputG}, \acts, \reps, \thresh, \wmax, 
\weightSet),
$$ where
\begin{itemize}
\item $\gens=\{\gene_1,\ldots,\gene_{\geneN}\}$ is a finite ordered
set of genes,
\item $\gene_{\inputG} \in \gens$ is the special gene used to model the constant input
for all genes,
\item $\acts \subseteq \gens \times \gens$ is the activation relation
such that $\forall \gene \in G\setminus\{\gene_{\inputG}\} \;(\gene_{\inputG},\gene)\in\acts$ and $\forall g\; (\gene,\gene_{\inputG})\notin \acts$,
\item $\reps \subseteq \gens \times \gens$ is the repression relation
such that $ \reps \cap \acts = \emptyset 
\land \forall g \;(\gene,\gene_{\inputG})\notin \reps$,
\item $\thresh \colon \gens \to \mathbb{Q}$ is the threshold
function such that 
$\forall \gene \in \gens\setminus\{\gene_{\inputG}\} \;\thresh(\gene) \geq 0$
and $\thresh(\gene_{\inputG})  < 0 $,
\item $\wmax \colon (\acts\cup\reps)\to \mathbb{Q}_{\geq 0}$ is the
maximum value of an activation/repression,
\item $\weightSet = \powerset((\acts\cup\reps)\to \mathbb{Q}_{\geq
0})$
assigns a set of
possible weight functions to each activation/inhibition relation, so that $\weight\in\weightSet\Rightarrow \forall(\gene,\gene')\in\acts\cup\reps\;\weight(\gene,\gene')\leq\wmax(\gene,\gene')$.
\end{itemize}
\end{defn}

In the following text, if not explicitly specified otherwise, we will be
referring to the \GRNlandscape\;
$\topp=(\gens,\gene_{\inputG},\acts,\reps, \thresh, \wmax, 
\weightSet)$.

\subsection{\DETGRN}

\begin{defn}[\DETGRN] \rm \label{def:grn}
A \DETGRN\;$\grn$ is a pair $(\topp,\weight)$, where
$\weight\in\weightSet$ is a weight function from the \GRNlandscape.
%
%
\end{defn}

A state $\state \colon \gens \to \B$ of a
\DETGRN\;$\grn = (\topp,\weight)$ denotes the activation state of each
gene in terms of a Boolean value.
Let $\states(\grn)$ (resp. $\states(\topp)$) denote the set of all
states of $\grn$ (resp. $\topp$), such that
$\state(\gene_{\inputG})=true$.
The GRN model executes in discrete time steps by updating all the
activation states synchronously and deterministically according to the
following rule: a gene is active at next time if and only if the total
influence on that gene, from genes active at current time, surpasses
its threshold.

\begin{defn}[Semantics of a \DETGRN] \rm \label{def:grn-sem}
A {\em run} of a \DETGRN\; $\grn= (\topp,\weight)$ 
is an infinite sequence of states $\state_0, \state_1, \dots$ such
that $\state_n \in \states(\grn)$ and $\trans(\state_n) =
\state_{n+1}$ for all $n \geq0$, where $\trans \colon \states(\grn)
\to \states(\grn)$ is a deterministic transition function defined by
\begin{align}
\trans(\state):= \lambda\gen'\hspace{-2pt}. 
\hspace{-2pt}\left[
\qquad
\sum_{\hspace{-11ex}\mathrlap{\{\gen \mid \state(\gen) \land (\gen,\gen') \in \acts\}}} 
\weight(\gen) \quad~-\quad
\sum_{\hspace{-7ex}\mathrlap{\{ \gen \mid \state(\gen) \land (\gen,\gen') \in \reps\}}}
\weight(\gen)
\qquad
 > \thresh(\gen')\right]\text{.}
\label{eq:trans}    
\end{align} 
The language of $\grn$, denoted by $\llbracket\grn\rrbracket$, is
a set of all possible runs of $\grn$.
Note that a \DETGRN\;does not specify the initial state.
Therefore, $\llbracket\grn\rrbracket$ may contain more than one run.
\end{defn}

%

\subsection{Temporal properties}
A GRN exists in a living organism to exhibit certain behaviors.
Here we present a linear temporal logic (LTL) to express the expected
behaviors of GRNs. 

\begin{defn}[Syntax of Linear Temporal properties]\label{def:ltl} \rm
The grammar of the language of temporal linear logic formulae are
given by the following rules
\begin{align*}
\prop &::=~ ~\genvar{\gen} ~\mid~ 
(\lnot \prop) ~\mid~ (\prop \lor \prop) ~\mid~ 
(\prop \ltluntil \prop)  \text{,}
\end{align*}
where $\gen \in \gens$ is a gene.
\end{defn}

Linear temporal properties are evaluated over all (in)finite runs of
states from $\states(\gens)$.
Let us consider a run $r = \state_1,\state_2,\state_3,\dots \in
\states(\gens)^{*} \cup \states(\gens)^{\infty}$.
Let $r^i$ be the suffix of $r$ after $i$ states and $r_i$ is the $i$th
state of $r$.
The satisfaction relation $\models$ between a run and an LTL formula is
defined as follows:
\begin{align*}
  &r \models g \text{ if } r_1(g), \qquad\;\;\;\; 
 r \models \lnot\prop \text{ if } r \not\models \prop,
\qquad
r \models \prop_1\lor\prop_2 \text{ if } r \models \prop_1 
\text{ or } r \models \prop_2,
\\
&r \models (\prop_1 \ltluntil \prop_2) \text{ if }
\exists i. r^i \models \prop_2\text{ and }\forall j \leq i. r^j \models \prop_1.
\end{align*} 
Note that if $|r| < i$ then $r^i$ has no meaning.
In such a situation, the above semantics returns {\em undefined}, i.e.,
$r$ is too short to decide the LTL formula. 
We say a language $\Lan \models \prop$ if for each run $r \in \Lan$,
$r \models \prop$, and a GRN $\grn \models \prop$ if
$\Lan\llbracket\grn\rrbracket\models \prop$. 
Let $\ltlfinally \prop$ be shorthand of 
$true \ltluntil \prop$ and $\ltlglobally \prop$ be shorthand of 
$\lnot\ltlfinally \lnot\prop$.

Note that we did not include next operator in the definition of LTL.
This is because a typical GRN does not expect something 
is to be done in strictly next cycle. 
%



%% file: tex_parts/algorithm.tex
In this section we present an algorithm for synthesizing the weights'
space corresponding to a given property in linear temporal logic.
The method combines LTL model
checking~\cite{PrinciplesOfModelChecking} and satisfiability modulo
theory (SMT) solving~\cite{barrett2009satisfiability}.

The method operates in two steps.
First, we represent any \DETGRN\; from the \GRNlandscape\; with a
parametrized transition system.
In this system, a transition exists between every two states, and it
is labelled by linear constraints, that are necessary and sufficient
constraints to enable that transition in a concrete \DETGRN\; (for
example, see \figref{fig:ab}b)). 
We say that a run of the parametrized transition system is {\em feasible}
if the conjunction of all the constraints labelled along the run is
satisfiable.
Second, we search for all the feasible runs that satisfy the desired LTL
property and we record the constraints collected along them.
The disjunction of such run constraints fully characterizes the
regions of weights which ensure that LTL property holds in the
respective \DETGRN.


%

\begin{defn}[Parametrized transition system] \rm
\label{def:pts}
For a given \GRNlandscape\; $\topp$ 
and rational parameters map $\wVar : \gens \to \allVars$, the {\em parametrized transition system}
$(\topp, \wVar )$ is a labelled transition system
$(\states(\topp), \transVar)$, 
where the labeling of edges
$\transVar : \states(\topp) \times \states(\topp) \to \polyhedra(\wVar)$ is defined as follows:\\
\begin{align*}
\transVar:= 
 \lambda \state \state'.\hspace{-1ex}
    \bigwedge_{\gen' \in \gens} 
    \hspace{-1ex}\left[
\quad
\sum_{\hspace{-8ex}\mathrlap{\{\gen \mid \state(\gen) \land (\gen,\gen') \in \acts\}}} 
\wVar(\gen) \quad~-\quad
\sum_{\hspace{-7ex}\mathrlap{\{ \gen \mid \state(\gen) \land (\gen,\gen') \in \reps\}}}
\wVar(\gen)
\;\;
>\;\; \thresh(\gen')
\iff \state'(\gen')
    \right]   
 \text{.}
\end{align*}  
\end{defn}
$\transVar(\state,\state')$ says that a gene $g'$ is active in
$\state'$ iff the weighted sum of activation and suppression activity of the
regulators of $g'$ is above its threshold.

A {\em run} of $(\topp, \wVar )$ is a sequence of states
$\state_0, \state_1, \dots$ such that $\state_n \in \states(\topp)$ for
all $n \geq0$, and $\transVar(\state_0,\state_{1}) \land
\transVar(\state_1,\state_{2}) \land \dots$ is said to be the {\em
run constraint} of the run.
A run is feasible if its run constraint is satisfiable.
We denote by $\llbracket(\topp,\wVar)\rrbracket$ the set of
feasible traces for $(\topp,\wVar)$.
For a weight function ${\weight}$, let $\transVar(\state,
\state')[{\weight}/\wVar]$ denote the formula obtained by
substituting $\wVar$ by ${\weight}$ and let $(\topp,
\wVar )[{\weight}/\wVar] = (\states(\topp), \transVar')$, where 
$\transVar'(\state,\state') = \transVar(\state,\state')[{\weight}/\wVar]$
for each $\state,\state' \in \states(\topp)$.

In the following text, we refer to the parametrized transition system
$(\topp, \wVar)$ and an LTL property $\prop$.  
Moreover, we denote the run constraint of run $\run=\state_0,\state_1,\ldots\in \llbracket(\topp,\wVar)\rrbracket$ by $\cons(\run)$.  

\begin{lemma}\label{lem:feasruns}\rm
  For a weight function ${\weight}$,
  the set of feasible runs of $(\topp, \wVar )[{\weight}/\wVar]$
  is equal to 
  $\llbracket(\topp,{\weight})\rrbracket$.
\end{lemma}

The proof of the above lemma follows from the definition of the semantics for \DETGRN.
Note that the run constraints are conjunctions of linear (non)-strict
inequalities.
Therefore, we may apply efficient SMT
solvers to analyze $(\topp, \wVar )$. 

\subsection{Constraint generation via model checking}

\input{./fig/fig-alg-model-check}

Now our goal is to synthesize the constraints over $\wVar$ which
characterise exactly the set of weight functions $\weight$, for which
$(\topp,{\weight})$ satisfies $\prop$.
%
Each feasible run violating $\prop$ reports a set of constraints which weight parameters should avoid.
Once all runs violating $\prop$ are accounted for, the desired region of weights is completely characterized.
More explicitly, the desired space of weights is obtained by conjuncting negations of run constraints of all feasible runs that satisfy $\neg\prop$.

In~\figref{fig:model-checking},
we present our algorithm \algGenCons for the constraint generation.
\algGenCons unfolds $(\topp, \wVar )$ in depth-first-order
manner to search for runs which satisfy $\neg\prop$.
At line 3, \algGenCons calls recursive function \algGenConsRec to do
the unfolding for each state in $\states(\topp)$.
\algGenConsRec takes six input parameters.
The parameter $run.\state$ and $runCons$ are the states of the
currently traced run and its respective run constraint.
The third parameter are the constraints, collected due to the
discovery of counterexamples, {\em i.e.}, runs which violate $\prop$.
The forth, fifth and sixth parameter are the description of the
input transition system and the LTL property $\prop$.
Since \DETGRN s have deterministic transitions, we only need to look
for the lassos upto length $|\states(\topp)|$ for LTL model checking.
Therefore, we execute the loop at line 7 only if the $run.\state$ has
length smaller than $|\states(\topp)|$.
The loop iterates over each state in $\states(\topp)$.
The condition at line 9 checks if $run.\state\state'$ is feasible
and, if it is not, the loop goes to another iteration.
Otherwise, the condition at line 10 checks if $run.\state\state'
\models \neg\prop$.
Note that $run.\state\state' \models \neg\prop $ may also return 
undefined because the run may be too short to decide the LTL property.
If the condition returns true, we add negation of the run constraint
in $goodCons$.
Otherwise, we make a recursive call to extend the run at line 13.
$goodCons$ tells us the set of values of $\wVar$ for which we have
discovered no counterexample.
\algGenCons returns $goodCons$ at the end.

Since run constraints are always a conjunction of linear inequalities,
$goodCons$ is a conjunction of clauses over linear inequalities.
Therefore, we can apply efficient SMT technology to evaluate the
condition at line 9.
The following theorem states that the algorithm \algGenCons computes
the parameter region which satisfies property $\varphi$. 

\begin{theorem}\rm
For every weight function $\weight\in\weightSet$, the desired set of
weight functions for which a \DETGRN\; satisfies $\prop$ equals the
weight functions which satisfy the constraints returned by
$\algGenCons$:
\[(\topp,\weight)\models\prop \hbox{ iff }
 \weight\models\algGenCons((\topp, v), \varphi).
\]
\end{theorem}

\begin{proof}
The statement amounts to showing that the sets  $A = \{\weight\mid (\topp,\weight)\models\prop\}$ and $B = \bigcap_{\run\in\llbracket(\topp,v)\rrbracket\wedge \run\models\neg\prop}\{\weight\mid \weight\models \neg\cons(\run)\}$ are equivalent.
Notice that
\begin{align*}
W\setminus A & = \{\weight\mid \exists \run\in\llbracket(\topp,v)\rrbracket\hbox{  such that } \weight\models  cons(\run)\wedge\run\models\neg\prop\} \\
& =  \bigcup_{\run\in\llbracket(\topp,v)\rrbracket\wedge\run\models\neg\prop}\{\weight\mid \weight\models cons(\run)\} \\
& = W \setminus \bigcap_{\run\in\llbracket(\topp,v)\rrbracket\wedge\run\models\neg\prop} \{\weight\mid \weight\models \neg cons(\run)\}.
\end{align*}
\end{proof}

We use the above presentation of the algorithm for easy readability.
However, our implementation of the above algorithm differs
significantly from the presentation.
We follow the encoding of~\cite{biere2003bounded} to encode the path
exploration as a bounded-model checking problem.
Further details about implementation are available in
Section~\ref{sec:experiments}.
The algorithm has exponential complexity in the size of $\topp$.
However, one may view the above procedure as the clause
learning in SMT solvers, where clauses are learnt when the LTL formula
is violated \cite{zhang2001efficient}.
Similar to SMT solvers, in practice,
this algorithm may not suffer from the worst-case complexity.

\begin{example}\rm
\label{ex:oscilations}
The GRN \texttt{osc3} (shown in ~\figref{fig:benchmarks}) was the model
of a pioneering work in synthetic biology \cite{elowitz2000synthetic},
and it is known to provide oscillatory behaviour: each gene should
alternate its expression between `on' and `off' state:
\begin{align*}
\prop_{3} = \bigwedge_{v\in\{A,B,C\}}(v \Rightarrow \ltlfinally\neg{v})
\wedge 
(\neg{v} \Rightarrow \ltlfinally {v}).
\end{align*}
The solutions are the constraints:
\begin{align*}
(\topp,\weight) \models \prop_{3} & \hbox{ iff } (i_A>\thresh_A)\wedge(i_B>\thresh_B)\wedge(i_C>\thresh_C) \wedge\\
 &(i_B-\weight_{AB}\leq\thresh_B) \wedge (i_C-\weight_{BC}\leq\thresh_C) \wedge (i_A-\weight_{CA}\leq\thresh_A).
\end{align*} 
\end{example}

%% file: fig/fig-alg-model-check.tex
\begin{figure}[t]
  \centering
  \begin{minipage}[t]{0.00\linewidth}
  \algLineNumbers{0pt}{2}{1}{5}    
  \algLineNumbers{0pt}{4}{6}{15}
  \end{minipage}
  \begin{minipage}[t]{0.99\linewidth}
    \algFunction $\algGenCons( ( \topp, \wVar ) = 
    (\states(\topp), \transVar), \prop ) $\\
    \algBegin \\
    \tabt $goodCons := true $ \\
    \tabt \algForeach $\state \in \states(\topp) $ \algDo \\
    \tabT $goodCons := \algGenConsRec( \state, true, goodCons, \states(\topp), 
    \transVar,\prop)$
    \\
    \tabt \algDone\\
    \tabt \algReturn $goodCons$\\
    \algEnd \\
    \\
    \algFunction $\algGenConsRec( run.\state, runCons, goodCons,\states(\topp), \transVar,\prop )$\\
    \algBegin \\
    \tabT \algIf  $|run.\state| < |\states(\topp)|$ \algThen \\
    \tabTT \algForeach $\state' \in \states(\topp) $ \algDo \\
    \tabTTT $runCons' := runCons \land \transVar(\state,\state')$\\
    \tabTTT \algIf $goodCons \wedge runCons'$ is sat \algThen \\
    \tabTTTT \algIf $run.\state\state' \models \neg\prop$ \algThen     \hfill\algComment{check may return undef}\\
    \tabTTTTT $goodCons := goodCons \land \neg runCons'$ 
    \\
    \tabTTTT \algElse \\
    \tabTTTTT $goodCons := \algGenConsRec( run.\state\state', runCons', goodCons, \states(\topp), \transVar, \prop )$ \\
    \tabTT \algScopeRule{2.65cm} \algDone\\
    \tabT \algReturn $goodCons$ \\
    \algEnd 
  \end{minipage}
  \caption{Counterexample guided computation of the mutation
space feasible wrt. $\prop$. Let ``." be an operator that appends two sequences. 
$run.\state\state' \models \neg\prop$ can be implemented by converting
$\neg\prop$ into a B\"uchi automaton and searching for an accepting run over
$run.\state\state'$. 
However, a finite path may be too short to decide whether $\prop$ holds or not.
In that case, the condition at line 10 fails.
Since $( \topp, \wVar )$ is finite, the finite runs are bound
to form lassos within $|\states(\topp)|$ steps.
If a finite run forms a lasso, then the truth value of
$run.\state\state' \models \neg\prop$ will be determined.
}
\label{fig:model-checking}
\end{figure}




%% file: tex_parts/robustness.tex
In this section, we present an application of our parameter synthesis
algorithm, namely computing robustness of GRNs in presence of
mutations.
%
%
To this end, we formalize \POPGRN\;and its robustness.
Then, we present a method to compute the robustness using our synthesized parameters.

A \POPGRN\; models a large number of \DETGRN s with varying weights. 
All the \DETGRN s are defined over the same \GRNlandscape, hence they
differ only in their weight functions.
The \POPGRN\; is characterised by the \GRNlandscape\; $\topp$ and a
probability distribution over the weight functions.
In the experimental section, we will use the range of weights $\weightSet$ and the distribution $\pi$ based on the mutation model outlined in the Appendix. 

\begin{defn}[\POPGRN] \rm \label{def:grnP}
A \POPGRN\; $\grnP$ is a pair $(\topp,\pi)$, where
$\pi:\weightSet\rightarrow[0,1]$ is a probability distribution over
all weight functions from \GRNlandscape\;$\topp$. 
\end{defn}

%


We write $\prop({\grnP})\in[0,1]$ to denote an expectation that a GRN
instantiated from a \POPGRN\;$\grnP=(\topp,\pi)$ satisfies
$\prop$.
The value $\prop({\grnP})$ is in the interval $[0,1]$ and we call it
robustness.

\begin{defn} [Robustness]\rm
\label{def:robustness}
Let $\grnP=(\topp,\pi)$ be a \POPGRN, and $\prop$ be an LTL
formula which expresses the desired LTL property.
Then, robustness of $\grnP$ with respect to property $\prop$ is given
by
\[
\prop( \grnP ):=\sum_{\{\weight\mid\llbracket(\topp,\weight) \rrbracket \models\prop \}}\pi(\weight)
\]
\end{defn} 

%
%
The above definition extends that of~\cite{Wagner2006_robustness}, because
it allows for expressing any LTL property as a phenotype, and hence it
can capture more complex properties such as oscillatory behaviour.

In the following, we will present an algorithm for computing the
robustness, which employs algorithm \algGenCons. 

\subsection{Evaluating robustness}
\label{sec:exact}

Let us suppose we get a \POPGRN\; $\grnP = ( \topp, \pi )$ and LTL
property $\prop$ as input to compute robustness.

For small size of \GRNlandscape\;$\topp$, robustness can be computed
by explicitly enumerating all the \DETGRN s from $\topp$, and
verifying each \DETGRN\; against $\prop$.
The probabilities of all satisfying \DETGRN s are added up.
However, the exhaustive enumeration of the \GRNlandscape\; is often
intractable due to a large range of weight functions $\weightSet$ in
$\topp$.
In those cases, the robustness is estimated statistically: a number of
\DETGRN s are sampled from $\topp$ according to the distribution
$\pi$, and the fraction of satisfying \DETGRN s is stored. 
The sampling experiment is repeated a number of times, and the mean
(respectively variance) of the stored values are reported as
robustness (respectively precision).
%
%

Depending on how a sampled \DETGRN\;is verified against the LTL
property, we have two methods:
\begin{itemize}
\item In the first method, which we will call \emph{execution} method,
each sampled \DETGRN\;is verified by executing the \DETGRN\; from all
initial states and checking if each run satisfies $\prop$;
\item In the second method, which we will call \emph{evaluation}
method, the constraints are first precomputed with $\algGenCons$, and
each sampled \DETGRN\; is verified by evaluating the constraints.
\end{itemize}

Clearly, the time of computing the constraints initially renders the
evaluation method less efficient.
This cost is amortized when verifying a \DETGRN\;by constraint
evaluation is faster than by execution.
In the experimental section, we compare the overall performance
of the two approaches on a number of GRNs from literature. 

%




%% file: tex_parts/experiments.tex
\input{./fig/fig-benchmark-grns.tex}

\newcommand{\sci}[2]{{#1\times 10^{#2}}}

We implemented a tool which synthesizes the parameter constraints for
a given LTL property (explained in \secref{sec:algo}), and the methods
for computing the mutational robustness (explained in
\secref{sec:exact}).
We ran our tool on a set of GRNs from literature.


{\subsection{Implementation}}
\label{sec:implementation}
Our implementation does not explicitly construct the parametrised
transition system described in \secref{sec:algo} (\dfnref{def:pts} and
Alg.~\ref{fig:model-checking}).
Instead, we encode the bounded model-checking
(Alg.~\ref{fig:model-checking}) as a satisfiability problem, and we
use an SMT solver to efficiently find $goodCons$. 
More concretely, we initially build a formula which encodes the
parametrized transition system and it is satisfied if and only if some
run of $( \topp, \wVar )$ satisfies $\neg\prop$.
If the encoding formula is satisfiable, the constraints $cons(\run)$
along the run are collected, and $\neg cons(\run)$ is added to
$goodCons$.
Then, we expand the encoding formula by adding $\neg cons(\run)$, so
as to rule out finding the same run again. 
We continue the search until no satisfying assignment of the encoding
formula can be found. 
The algorithm always terminates because the validity of the property
is always decided on finite runs (as explained in \secref{sec:algo}).

The implementation involves 8k lines of C++ and we use
Z3 SMT solver as the solving engine.
We use CUDD to reduce the size of the Boolean structure of the
resulting formula. 
We ran the experiments on a GridEngine managed cluster system. 
Our tool, as well as the examples and the simulation results are
available online.\footnote{\url{http://pub.ist.ac.at/~mgiacobbe/grnmc.tar.gz}}.

\input{./fig/tab-benchmark-2.tex}
%
\subsection{Performance evaluation}
We chose eight GRN topologies as benchmarks for our tool.
The benchmarks are presented in~\figref{fig:benchmarks}.
The first five of the GRN topologies are collected
from~\cite{cardellimorphisms}.
On these benchmarks we check for the steady-state properties.
On the final three GRN topologies, we check for the oscillatory behavior.
%
%
The results are presented in \figref{tab:time}.

We ran the robustness computation by the evaluation and execution
methods (the methods are described in \secref{sec:exact}).
In order to obtain robustness and to estimate the precision, we
computed the mean of $100$ experiments, each containing a number of
samples ranging from $10^3$ to $10^6$.
The total computation time in the execution methods linearly depends
on the number of samples used.
The total computation time in the evaluation method depends linearly
on the number of samples, but initially needs time to compute the
constraints.
Technically, the time needed to compute robustness by execution method
is $t_{ex} = k_{ex}p$, and the time needed to compute robustness by
evaluation approach $t_{ev} = k_{ev}p+t_{c}$, where $p$ represents the
total number of samples used, $t_c$ is the time to compute the
constraints, and $k_{ex}$ (resp. $k_{ev}$) is the time needed to
verify the property by evaluation (resp. execution).
 We used linear regression to estimate the parameters $k_{ex}$ and
$k_{ev}$, and we present the ratio $\frac{k_{ex}}{k_{ev}}$ in top-left
position of \figref{tab:time}.
The results indicate that on six out of eight tested networks,
evaluation is more efficient than execution.
For some networks, such as \texttt{osc7}, the time for computing the
constraints is large, and the gain in performance becomes visible only
once the number of samples is larger than $10^6$.

%% file: fig/fig-benchmark-grns.tex
\begin{figure}[t]
  \centering
    \begin{tabular}[t]{r@{\quad}c@{\quad}l@{\quad}l}
      & & \multicolumn{1}{c}{Property} & \multicolumn{1}{c}{Space size} \\
      \texttt{mi:} &
\begin{tabular}{l}
\begin{tikzpicture}[scale=0.6,->,>=stealth',
  shorten >=1pt,auto,node distance=1.7cm,
  thick,main node/.style={circle,draw,transform shape}]

  \node[main node] (A) {Y};
  \node[main node] (B) [right of=A] {Z};

  \path[every node/.style={font=\sffamily\small}]
    (A) edge [bend left,-|] (B)
    (B) edge [bend left,-|]  (A);
\end{tikzpicture}
\end{tabular}
& 
$(Y\bar{Z} \hspace{-4pt}\implies\hspace{-4pt} \ltlglobally Y\bar{Z}) 
\land ( Z\bar{Y} \hspace{-4pt}\implies\hspace{-4pt} \ltlglobally Z\bar{Y})$
&$4225$
\\
\texttt{misa:} &
\begin{tabular}{l}
\begin{tikzpicture}[scale=0.6,->,>=stealth',shorten >=1pt,auto,
  node distance=1.7cm,
  thick,main node/.style={circle,draw,transform shape}]

  \node[main node] (A) {Y};
  \node[main node] (B) [right of=A] {Z};
  \path[every node/.style={font=\sffamily\small}]
    (A) edge [bend left,-|] (B)
        edge [loop left]  (A)
    (B) edge [bend left,-|]  (A)
        edge [loop right] (B)
        ;
\end{tikzpicture}
\end{tabular}
&
$(Y\bar{Z}\hspace{-4pt}\implies\hspace{-4pt} \ltlglobally Y\bar{Z}) 
\land ( Z\bar{Y}\hspace{-4pt}\implies\hspace{-4pt}\ltlglobally Z\bar{Y})$
&
$105625$
\\
\texttt{qi:} &
\begin{tabular}{l}
\begin{tikzpicture}[scale=0.6,->,>=stealth',shorten >=1pt,auto,node distance=1.7cm,
  thick,main node/.style={circle,draw,transform shape}]

  \node[main node] (A) {Y};
  \node[main node] (B) [right of=A] {Z};
  \node[main node] (C) [below of=B] {R};
  \node[main node] (D) [left  of=C] {S};

  \path[every node/.style={font=\sffamily\small}]
    (A) edge [bend left,-|] (B) edge [bend left   ] (D)
    (B) edge [bend left   ] (C) edge [bend left,-|] (A)
    (C) edge [bend left,-|] (D) edge [bend left   ] (B)
    (D) edge [bend left   ] (A) edge [bend left,-|] (C)
        ;
\end{tikzpicture}
\end{tabular}
&
\begin{tabular}{l}
$(YS\bar{Z}\bar{R} \implies \ltlglobally YS\bar{Z}\bar{R})\hspace{4pt} \land $ \\
$(\bar{Y}\bar{S}ZR \implies \ltlglobally \bar{Y}\bar{S}ZR) $
\end{tabular}
&  $\approx10^9$
\\
\texttt{cc:}
&
\begin{tabular}{l}
\begin{tikzpicture}[scale=0.6,->,>=stealth',shorten >=1pt,
  auto,node distance=1.7cm,
  thick,main node/.style={circle,draw,transform shape}]

  \node[main node] (A) at (-2,0) {S};
  \node[main node] (B) at (-1,0) {Z};
  \node[main node] (C) at (0,1) {X};
  \node[main node] (D) at (1,0) {R};
  \node[main node] (E) at (2,0) {T};

  \path[every node/.style={font=\sffamily\small}]
    (A) edge (B)
    (B) edge [bend left,-|] (C)
    (C) edge [bend left,-|] (B) edge [bend right   ] (D)
    (D) edge [bend right   ] (C)
    (E) edge [-|] (D)
        ;
\end{tikzpicture}
\end{tabular}
& 
$\ltlfinally \ltlglobally X  \lor \ltlfinally \ltlglobally \bar{X}$
&  $\approx 10^{10}$
\\
\texttt{ncc:}
&
\begin{tabular}{l}
\begin{tikzpicture}[scale=0.6,->,>=stealth',shorten >=1pt,auto,node distance=1.7cm,
  thick,main node/.style={circle,draw,transform shape}]

  \node[main node] (Q) at (-1,1)  {Q};
  \node[main node] (S) at (1,1)  {S};
  \node[main node] (Z) at (2,0)  {Z};
  \node[main node] (R) at (1,-1) {R};
  \node[main node] (P) at (-1,-1) {P};
  \node[main node] (Y) at (-2,0) {Y};

  \path[every node/.style={font=\sffamily\small}]
    (Q) edge [bend right,->] (Y)
    (Y) edge [bend right,->] (Q)
    (Y) edge [bend right=18,->] (S)
    (Y) edge [bend left=18,-|] (R)
    (Z) edge [bend left=18, -|] (Q)
    (Z) edge [bend right=18, ->] (P)
    (Z) edge [bend left, -|] (S)
    (S) edge [bend left, -|] (Z)
    (Z) edge [bend right, ->] (R)
    (R) edge [bend right, ->] (Z)
    (P) edge [bend left, -|] (Y)
    (Y) edge [bend left, -|] (P);
\end{tikzpicture}
\end{tabular}
& 
\begin{tabular}{l}
$(YQS\bar{Z}\bar{P}\bar{R}\hspace{-4pt}\implies\hspace{-4pt}\ltlglobally YQS\bar{Z}\bar{P}\bar{R})\hspace{4pt} \land$\\
$(\bar{Y}\bar{Q}\bar{S}ZPR\hspace{-4pt}\implies\hspace{-4pt}\ltlglobally \bar{Y}\bar{Q}\bar{S}ZPR)$
\end{tabular}
&  $\approx10^{18}$
\\
\texttt{osc}N\texttt{:}
&
\begin{tabular}{l}
\begin{tikzpicture}[scale=0.6,->,>=stealth',shorten >=1pt,auto,node distance=1.7cm,
  thick,main node/.style={circle,draw,transform shape}]

  \node[main node] (A) {$X_N$};
  \node[main node] (B) at ([shift={(120:1.7)}]A) {$X_1$};
  \node[main node] (C) at ([shift={(240:1.7)}]B) {$X_2$};

  \path[every node/.style={font=\sffamily\small}]
    (A) edge [bend right,-|] (B)
    (B) edge [bend right,-|] (C)
    (C) edge [dashed, bend right,-|] (A)
        ;
\end{tikzpicture}
\end{tabular}
&
\begin{tabular}{l}
  $X_1\hspace{-4pt}\implies\hspace{-4pt}\ltlfinally\bar{X_1} \land \bar{X_1}\hspace{-4pt}\implies\hspace{-4pt}\ltlfinally X_1 \hspace{8pt} \land$\\
  $X_2\hspace{-4pt}\implies\hspace{-4pt}\ltlfinally\bar{X_2} \land \bar{X_2}\hspace{-4pt}\implies\hspace{-4pt}\ltlfinally X_2 \hspace{8pt} \land $\\
  {\centering \dots}\\
  $X_N\hspace{-4pt}\implies\hspace{-4pt}\ltlfinally\bar{X}_N \land \bar{X}_N\hspace{-4pt}\implies\hspace{-4pt}\ltlfinally X_N $
\end{tabular}
&
\begin{tabular}{l}
 \texttt{3:}$274625$ \\
 \texttt{5:}$\approx10^9$\\
 \texttt{7:}$\approx10^{12}$
\end{tabular}
\\
\end{tabular}
 \\








  \caption{GRN benchmarks. 
\texttt{mi}, \texttt{misa} (mutual inhibition), \texttt{qi} (quad inhibition),
and \texttt{ncc} (cell cycle switch) satisfy different forms of bistability.
For the networks \texttt{ci} (cell cycle switch), the value of gene
eventually stabilizes~\cite{cardelli2012cell}.
In \texttt{osc3}, also known as the \emph{repressilator}
\cite{elowitz2000synthetic}, the gene values alternate.
\texttt{osc5} and \texttt{osc7} (not shown) are generalizations of
\texttt{osc3}, and also exhibit oscilating behavior.
}
  \label{fig:benchmarks}
\end{figure}



%% file: fig/tab-benchmark-2.tex
\begin{figure}[tp]
\centering
\includegraphics[width=0.32\textwidth]{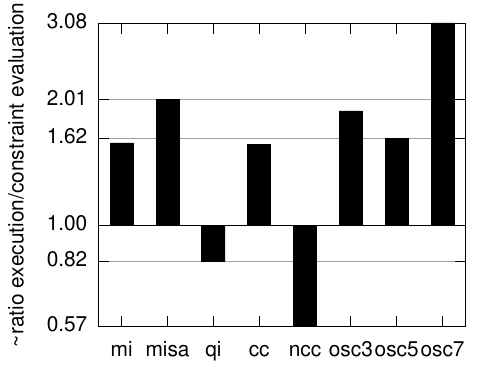}
\includegraphics[width=0.32\textwidth]{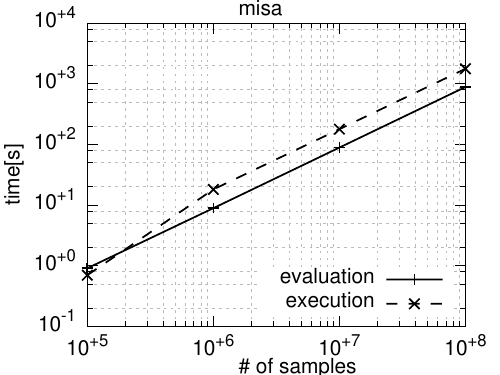}
\includegraphics[width=0.32\textwidth]{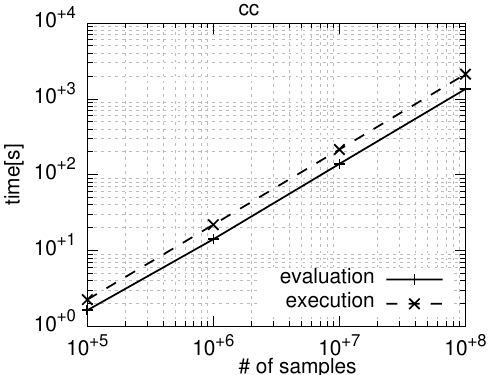}
\includegraphics[width=0.32\textwidth]{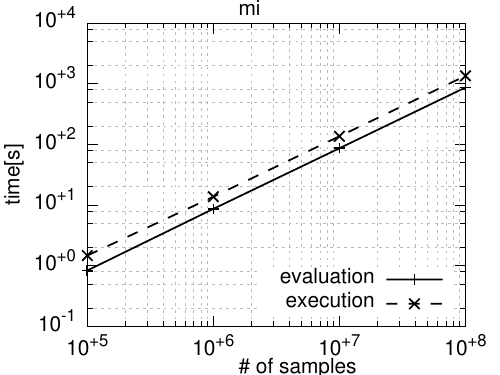}
\includegraphics[width=0.32\textwidth]{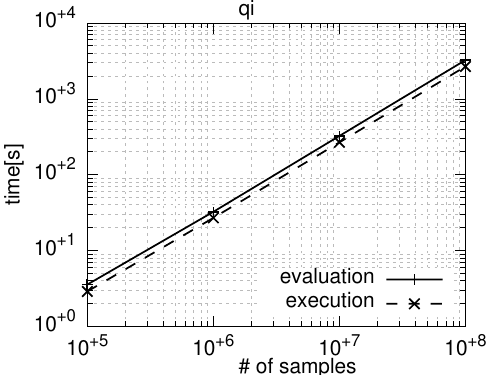}
\includegraphics[width=0.32\textwidth]{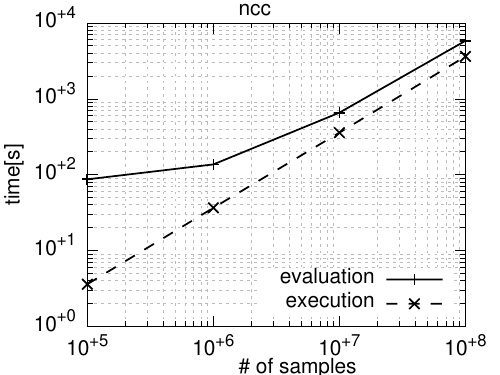}
\includegraphics[width=0.32\textwidth]{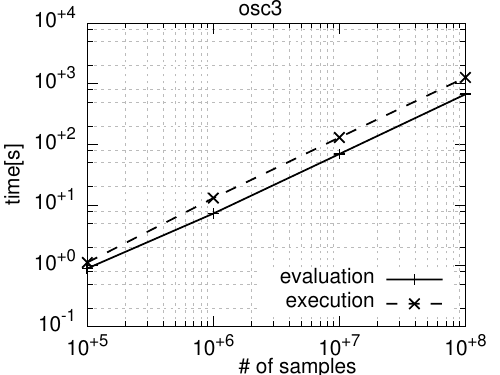}
\includegraphics[width=0.32\textwidth]{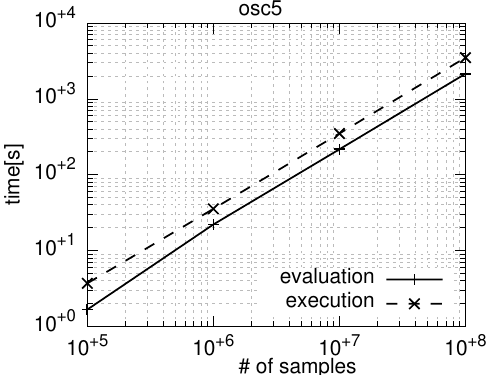}
\includegraphics[width=0.32\textwidth]{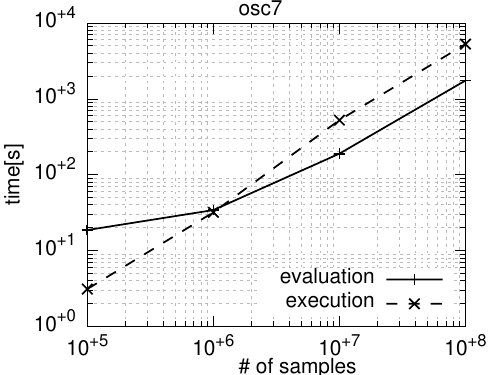}
\caption{
The comparison in performance when mutational robustness is statistically estimated, and a property check is performed either by evaluation, or by execution (see \secref{sec:exact} for the description of the two methods).
The bar (top-left) shows the ratio of average times needed to verify the property of one sampled \DETGRN.
For example, for \texttt{osc7}, the performance of evaluation method is more than three times faster.
%
%
The other graphs show how the robustness computation time depends on the total number of sampled GRNs (in order to obtain robustness and to estimate the precision, we computed the mean of $100$ experiments, each containing a number of samples ranging from $10^3$ to $10^6$). 
The graph is shown in log-log scale.
The non-linear slope of the evaluation method is most notable in examples \texttt{mcc} and \texttt{osc7}, and it is due to the longer time used for compute the constraints.
%
%
}
\label{tab:time}
\end{figure}
%



%% file: tex_parts/conclusion.tex
In this paper, we pursued formal analysis of Wagner's GRN model, which
allows symbolic reasoning about the behavior of GRNs under parameter
perturbations.
More precisely, for a given space of GRNs and a property specified in LTL, we have synthesized the space of parameters for which the concrete, individual GRN from a given space satisfies the property.
The resulting space of parameters is represented by complex linear inequalities.
In our analysis, we have encoded a bounded model-checking search into a satisfiability problem, and we used efficient SMT solvers to find the desired constraints. 
We demonstrated that these constraints can be used to efficiently compute the mutational robustness of populations of GRNs. 
Our results have shown the cases in which the computation can be three times faster than the standard (simulation) techniques employed in computational biology.

While computing mutational robustness is one of the applications of our synthesized constraints, 
the constraints allow to efficiently answer many other questions that are very difficult or impossible to answer by executing the sampled GRNs.
In our future work, we aim to work on further applications of our method, such as parameter sensitivity analysis for Wagner's model.
Moreover, we plan to work on the method for exact computation of robustness by applying 
the point counting algorithm~\cite{barvinok1999algorithmic}.

The Wagner's model of GRN is maybe the simplest dynamical model of a GRN 
-- there are many ways to add expressiveness to it: for example, by incorporating 
multi-state expression level of genes, non-determinism, asynchronous
updates, stochasticity. 
We are planning to study these variations and chart the territory of
applicability of our method.
%





%% file: tex_parts/mutation-model.tex
\section*{Mutation model} 
\label{sec:bio-mutation-model}

In this section, we present the mutation model that is grounded in evolutionary biology research.
%
Genetic mutations refer to events of changing base pairs in the
DNA sequence of the genes and they may disturb the regular
functioning of the host cell.
Such mutations affect the weight functions and we assume their range to be between
a maximum weight and zero. 

\subsubsection{Modeling mutations in a nucleotide.}
\label{sec:single}
The mutations of a single nucleotide follow a discrete-time Markov
chain (DTMC), illustrated in \figref{fig:a}a).
When the DNA is passed from the mother to the daughter cell, a
`correct' nucleotide (in figure it is the $\ade$) can mutate to each
different value ($\tim$, $\cyt$ or $\gua$) with some probability. 
%
%
In the model in \figref{fig:a}a), the probability of each possible
mutation is $\frac{p}{3}$, and hence, the probability of retaining the
same nucleotide is $1-p$.
In \figref{fig:a}b), there is a process where all mutated states are
lumped together - the probability to get to the `correct' nucleotide
is equal to $\frac{p}{3}$, and to remain mutated is therefore
$(1-\frac{p}{3})$.
%
%
We will refer to the general probabilities of the lumped process with
a matrix
\[ P^{id}=\left[ \begin{array}{ccc}
p_{00} & p_{01}  \\
p_{10} &p_{11}  
\end{array} \right],
\] 
with the intuition $0$ being the non-mutated state and $1$ the
mutated state.
Let $X_n^i\in\{\zero,\one\}$ be a process (DTMC) reporting whether the
$i$-th nucleotide is in its `correct' configuration, or mutated
configuration at $n$-th generation.
The number $P(X_n^i=\one)$ can be interpreted also as that the
fraction of mutated $i$-th nucleotides observed in the population at
generation $n$. 
Notably, for the values shown in \figref{fig:a}b), the stationary
distribution of $X_n^i$ is $(0.75,0.25)$, independently of the
probability of a single mutation $p$.

\begin{figure}[t]
\begin{center}
\includegraphics[scale=0.35]{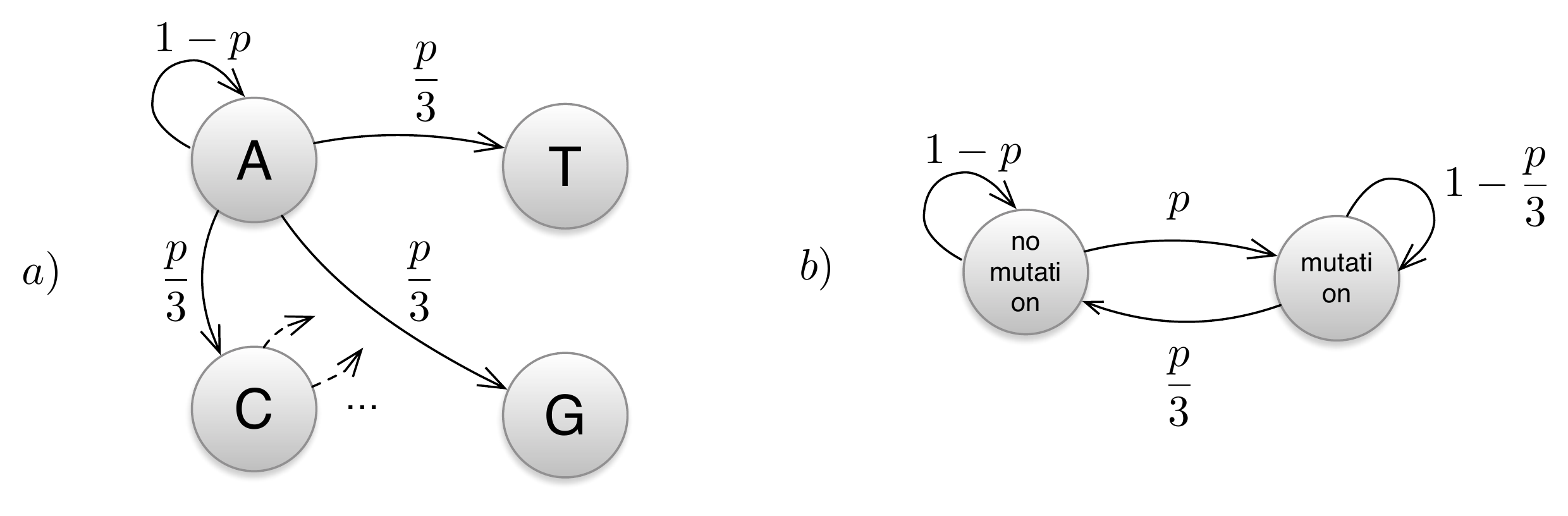}
\end{center}
\caption{A model of changes of a single nucleotide from one generation to another}
\label{fig:a}
\end{figure}

\subsubsection{Modeling mutations in a single gene.}  
Let us assume $\length$ to be the length of the DNA binding region of a gene
\footnote{More precisely the length of the promoter of gene} 
$\gene$, and $P^{id}$ to
be specify the probabilities of single point mutations occurring in
it. 
Moreover, we assume that the maximal weight for influence of gene $\gene$ to $\gene'$, {\em
i.e.} $\weight(\gene,\gene')$, is achieved for a single sequence of
nucleotides, and that the weight will linearly decrease with the
number of mutated nucleotides (where `mutated' refers to being
different to the sequence achieving the maximal weight).
Therefore, if $k$ out of $\length$ nucleotides in the promoter region of
$\gene'$ that is affected by $\gene$ are mutated, the weight $\weight(\gene,\gene')$ becomes
$\wmax(\gene,\gene')(1-\frac{k}{\length})$. 

If the whole promoter sequence of one gene is modelled by a string
$\avec\in\{\ade,\tim,\cyt,\gua\}^{\length}$, its changes over
generations are captured by an $l$-dimensional random process
$(X_n^1,\ldots,X_n^l)(\avec)\in\{\zero,\one\}^l$. 
Mutation events $X_n^1$, $X_n^2$,\ldots, $X_n^l$ are assumed to happen
independently within the genome, and independently of the history of
mutations. 
Then, a random process
$\K_n:=X_n^1+\ldots+X_n^l$,
 such that $\K_n=k\in\{0,\ldots,\length\}$, denotes
that the configuration at generation $n$ differs from the optimal
configuration in exactly $k$ points. 
The following lemma defines the values of the transition matrix of the process $\{\K_n\}$.

\begin{lemma} \rm
\label{lem:evolutionary}
The process $\{\K_n\}$ is a DTMC, with transition probabilities
$P(\K_{n+1}=j\mid \K_n=i) = \Q({i,j})$, where
$\Q: 0..l\times 0..l\to[0,1]$ amounts to:
\begin{align}
\label{eq:prob}
\Q(i,j) := \sum_{u=0}^{\min\{i,j\}} {i\choose u}p_{11}^u
p_{10}^{i-u}{{l-i}\choose{j-u}} p_{01}^{j-u}p_{00}^{l-i-(j-u)}.
\end{align}

Moreover, 
$\{\K_n\}$ converges to a unique stationary distribution, which is a binomial, with success probability
$\beta=\lim_{n\to\infty} P(X_n^i=\one)$.
In a special case
$\hat{P}^{id}=\left[ \begin{array}{ccc} 1-p & {p} \\
\frac{p}{3} & 1-\frac{p}{3}
\end{array} \right]$, 
the stationary distribution is $\beta=\frac{3}{4}$, independently of
the value of $p$.
\end{lemma}

\begin{proof}\rm(\lemref{lem:evolutionary})
It suffices to observe that $u$ represents the number of mutated nucleotides which remain mutated at time $(n+1)$, and $(j-u)$ is the number of unmated nucleotides which become mutated at time $(n+1)$.
The existence and convergence of the stationary trivially follows from that $\{\K_n\}$ is regular (irreducible and aperiodic).
Each process $\{X_n^i\}$ at a stationary behaves as a Bernoulli process, with probability of being mutated equal to $\beta=\lim_{n\to\infty} P(X_n^i=1)$, and a probability of remaining un-mutated $1-\beta$.
The claim follows as the sum of $l$ independent Bernoulli processes is binomially distributed.
The special case follows because matrix $\hat{P}^{id}$ has a unique stationary distribution $(0.25,0.75)$ (denoting by $a$ and $b$ its coordinates, we obtain equations $a(1-p)+b\frac{p}{3}=a$ and $ap+b(1-\frac{p}{3})=b$, which can be satisfied only if $a=\frac{b}{3}$).
\end{proof}

Remark that in a special case when $p_{00}=0.25$, 
the transition matrix of the process $\{\K_t\}$ has the transition probabilities  
 \begin{align*}
P(K_{n+1}=j\mid K_n=i) & = p_{11}^{j}p_{00}^{l-j} \sum_{u=0}^i {i\choose u}{{l-i}\choose{j-u}} \\
 & = {l\choose j}0.75^{j}0.25^{l-j},
 \end{align*}
which do not depend on $i$.

The special case of matrix $\Q$ for $l=2$ is illustrated in \figref{fig:evolutionary}.

\input{./fig/fig-evolutionary.tex}

\subsubsection{Modeling mutations in a GRN.}

We assume that mutations happen independently across the genome, and therefore,
if the genome is in the configuration $\kvec=(k_1,\ldots,k_d)$, the
probability that the next generation will be in the configuration
$\kvec'=(k_1',\ldots,k_d')$ will be a product of transition probabilities from 
$k_i$ to $k_i'$, for $i=1,\ldots,d$.

%
Each sequence $\kvec=(k_1,\ldots,k_d)\in
0..\length_1\times\ldots\times 0..\length_{\geneN}$ defines one weight
function $\weight_{\kvec}$, with
$\weight_{\kvec}(\gene_i,\gene_j)=\wmax(\gene_i,\gene_j)(1-\frac{k_j}{l_j})$. 
%
%
Hence, the domain of weight functions is given by
\[
\weightSet = \{\weight_{\kVec}\mid
\kVec\in0..\length_1\times\ldots\times 0..\length_{\geneN} 
\},
\]
and the distribution is given by
\[
\pi(\weight_{\kVec}) = \prod_{i=1}^d {d\choose k_i}\beta^{k_i}(1-\beta)^{d-k_i}. 
\]

%

%% file: fig/fig-evolutionary.tex
\begin{figure}
\begin{center}
\includegraphics[scale=0.5]{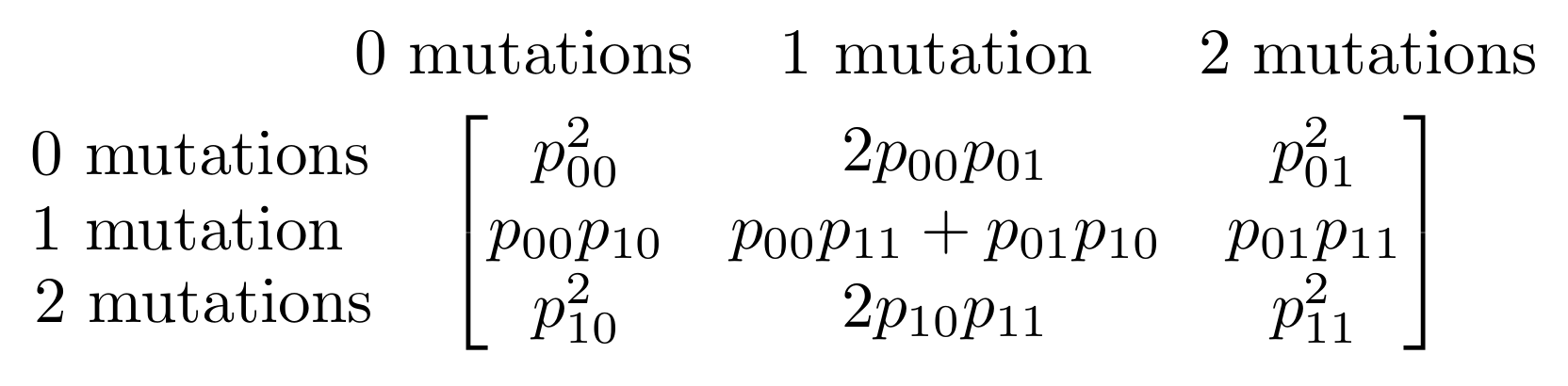}
\end{center}
\caption{The transition matrix for the evolutionary process of one gene with two nucleotides. 
The transition probabilities are computed based on the assumption that the single nucleotide mutations happen independently across the genome: for example, the transition from zero to two mutation is equal to $p_{01}^2$, 
because it occurs iff both nucleotides transition from a correct to a mutated state.
}
\label{fig:evolutionary}
\end{figure}